\documentclass{amsart}
\usepackage{graphicx}
\newtheorem{theorem}{Theorem}
\newtheorem{corollary}{Corollary}
\UseRawInputEncoding
\usepackage[russian, english]{babel}

\def\dfrac#1#2{\displaystyle{#1\over #2}}

\def\bv{{\bf v}}
\def\bp{{\bf p}}

\def\Div{\mbox{div}\,}

\def\bB{{\bf B}}

\def\bE{{\bf E}}

\begin{document}

\markboth{Rozanova, Chizhonkov}{On the conditions for the breaking of
oscillations in a cold plasma}

%
%

\title{On the conditions for the breaking of
oscillations in a cold plasma}

\author{Olga S. Rozanova}

\address{ Mathematics and Mechanics Department, Lomonosov Moscow State University, Leninskie Gory,
Moscow, 119991,
Russian Federation,
rozanova@mech.math.msu.su}

\author{Eugeniy V. Chizhonkov}

\address{ Mathematics and Mechanics Department, Lomonosov Moscow State University, Leninskie Gory,
Moscow, 119991,
Russian Federation,
chizhonk@hotmail.com}



\maketitle


\begin{abstract}
The Cauchy problem for a quasilinear system of hyperbolic equations describing plane one-dimensional relativistic oscillations of electrons in a cold plasma is considered. For some simplified formulation of the problem, a criterion for the existence of a  global in time solutions is obtained. For the original problem, a sufficient condition for the loss of smoothness is found, as well as a sufficient condition for the solution to remain smooth at least for time $ 2 \pi $. In addition, it is shown that in the general case, arbitrarily small perturbations of the trivial state lead to the formation of singularities in a finite time. It is further proved that there are special initial data such that the respective solution remains smooth for all time, even in the relativistic case. Periodic in space traveling wave gives an example of such a solution. In order for such a wave to be smooth, the velocity of the wave must be greater than a certain constant that depends on the initial data. Nevertheless,  arbitrary small perturbation of general form destroys these global in time smooth solutions.  The nature of the singularities of the solutions is illustrated by numerical examples.

\end{abstract}



\section*{Introduction}	

The hydrodynamic model of "cold" plasma is well known and described in detail in textbooks and monographs on plasma physics (e.g.\cite{GR75}). It is one of the simplest models in which a plasma is considered as a relativistic electron liquid, neglecting collisional and recombination effects, as well as the motion of ions. Physicists have long known that the waves described by such a model tend to break down. This phenomenon is associated with the release of energy and subsequent "heating" of the plasma, which means the inapplicability of the original model after the moment of breaking. There are many works devoted to the breaking effect in different versions of the model with an approximate analysis of this phenomenon, written at the physical level of rigor (in particular,  \cite {WPN92,Trines09,DL16} and references therein).

Currently, attention to this model is due, first of all, to the problems associated with the propagation of
superpower short laser pulses in a plasma \cite{ESL09,BEH16}. During the movement, the pulse excites behind the wake plasma wave, which is used to create special laser accelerators \cite{AK75}. Electrons there can be accelerated to high energies at substantially shorter distances than in traditional devices. A wake wave also has the ability to break down, transferring its energy to plasma particles. From a technical point of view, it is important to find the conditions under which the wake wave does not collapse for as long as possible.

In mathematical modeling of processes in a collisionless cold plasma, the Lagrangian or Euler approaches are used. In the first case, individual particle trajectories are tracked, and the second one is associated with a hydrodynamic description based on partial differential equations (see, for example, \cite{BL85,DnK}). The break down of oscillation signifies the intersection of electron trajectories in the first case, and  the blow up of the electron density in the second case \cite{CH18}.
From a mathematical point of view, this phenomenon means the appearance of a strong (delta-shaped) singularity in the function of  density of charge. A similar phenomenon occurs, for example, in "pressureless" gas dynamics.

In full, three-dimensional form, the equations of hydrodynamics of a "cold" plasma can be treated only numerically. In this paper, we  study a one-dimensional analogue of this model. This simplification allows a rigorous mathematical study of the breakdown phenomenon.

 We consider two different statements of the Cauchy problem for equations describing plane one-dimensional
electronic oscillations in a cold plasma. In one of them, non-relativistic, there is a wider possibility of the existence of a global smooth solution. Namely, the result of this paper allows us to separate the initial data of the Cauchy problem, which describes flat one-dimensional non-relativistic electron oscillations in a cold plasma, into two classes. One class generates a smooth $ 2 \pi$ - periodic solution on the infinite time interval, and the other class leads to the formation of a singularity during the first period of oscillation. For the relativistic statement, in the general case, the initial data of the Cauchy problem can be divided into two classes in a similar way only for a period of time corresponding to the first oscillation.  However, under a special condition on the initial data, one can obtain a criterion for the formation of a singularity and thereby find a class of smooth solutions.

The paper is organized as follows. In Sec.\ref{Sec1} we set the problem of plasma oscillations in the general three-dimensional case and reduce the complete system to the case of one-dimensional plane oscillations both with and without the relativistic effect. In Sec.\ref{Sec2}, a complete analysis of the Cauchy problem for the non-relativistic case is carried out and a criterion for the loss of smoothness in terms of the initial data is obtained. Sec.\ref{Sec3} is devoted to the relativistic case. First, a special situation is considered when the initial data are coupled. In this case, we obtain a criterion  for the formation of  singularity of a smooth solution in terms of the initial data. Further, the situation of arbitrary initial data is considered and we obtain a sufficient condition for preserving the smoothness during the first period of oscillations. This condition does not prohibit breaking after a sufficiently large number of oscillations. Further, we prove a criterion  for the existence of a globally smooth solution in terms of the Hill equation. Then we get  a corollary, which states that for uncoupled initial data, any arbitrarily small perturbation of the trivial solution leads to a finite-time appearance of a singularity. Sec.\ref{Sec4} discusses special class of solutions in the form of a traveling wave. These are functions periodic in space, whose period is the longer, the greater the speed of the wave. A restriction on the wave velocity is obtained, which ensures the global smoothness of the solution in time. Such solutions can have any amplitude, and, in particular, can be arbitrarily small perturbations of the trivial solution. However, they are unstable in the sense that they lose their smoothness when exposed to arbitrarily small disturbances. Sec.\ref{Sec5} presents numerical results that demonstrate the phenomenon of the formation of singularities by the example of typical initial data, the most interesting from the point of view of applications. In particular, it was shown that the nature of the emerging singularities is the same as that of a traveling wave with a low propagation velocity.

\section{Statement of the problem of plasma oscillations}\label{Sec1}

The system of equations of hydrodynamics of a ``cold'' plasma, including hydrodynamic equations together with Maxwell's equations in vector form, has the form (see, for example, \cite {S71,ABR78,SR12}):
\begin{equation}
\label{base1}
\begin{array}{l}
\dfrac{\partial n }{\partial t} + \Div(n \bv)=0\,,\quad
\dfrac{\partial \bp }{\partial t} + \left( \bv \cdot \nabla \right) \bp
= e \, \left( \bE + \dfrac{1}{c} \left[\bv \times  \bB\right]\right),\vspace{0.5em}\\
\gamma = \sqrt{ 1+ \dfrac{|\bp|^2}{m^2c^2} }\,,\quad
\bv = \dfrac{\bp}{m \gamma}\,,\vspace{0.5em}\\
\dfrac1{c} \frac{\partial \bE }{\partial t} = - \dfrac{4 \pi}{c} e n \bv
 + {\rm rot}\, \bB\,,\quad
\dfrac1{c} \frac{\partial \bB }{\partial t}  =
 - {\rm rot}\, \bE\,, \quad \Div \bB=0\,,
\end{array}
\end{equation}
where
$ e, m $ is the charge and mass of the electron (here the charge of the electron has a negative sign: $ e <0 $),
$ c $  is the speed of light, $ n, \bp, \bv $ is the concentration, momentum and velocity of electrons; $ \gamma $ is the Lorentz factor; $ \bE, \bB $ are the vectors of the electric and magnetic fields.

In order to analyze flat one-dimensional relativistic plasma oscillations, the basic equations ~(\ref{base1}) can be significantly simplified.

We will denote the independent variables in the Cartesian coordinate system as $ (x, y, z) $, and apply the assumptions that
\begin{itemize}

\item the solution is determined only by the $ x- $ components of the vector functions $ {\bp}, {\bv}, {\bE}; $

\item there is no dependence in these functions on the variables $ y $ and $ z $, i.e. $ \partial / \partial y = \partial / \partial z = 0. $
\end{itemize}

  Then from (\ref {base1}) we get
\begin{equation}
\begin{array}{c}
\dfrac{\partial n }{\partial t} +
\dfrac{\partial }{\partial x}
\left(n\, v_x \right)
=0,\quad
\dfrac{\partial p_x }{\partial t} + v_x \dfrac{\partial p_x}{\partial x}= e \,E_x,
\vspace{1 ex}\\
\gamma = \sqrt{ 1+ \dfrac{p_x^2}{m^2c^2}}\,, \quad
{v_x} = \dfrac{p_x}{m \,\gamma}, \quad
\dfrac{\partial E_x }{\partial t} = - 4 \,\pi \,e \,n\, v_x\,.
\end{array}
\label{3gl2}
\end{equation}

We introduce dimensionless quantities
$$
\rho = k_p x, \quad \theta = \omega_p t, \quad
V = \dfrac{v_x}{c}, \quad
P = \dfrac{p_x}{m\,c}, \quad
E = -\,\dfrac{e\,E_x}{m\,c\,\omega_p}, \quad
N = \dfrac{n}{n_0},
$$
 where $\omega_p = \left(4 \pi e^2n_0/m\right)^{1/2}$ is the plasma frequency, $n_0$ is the value of the unperturbed electron density, $k_p = \omega_p /c$.

In the new variables,  system (\ref {3gl2}) takes the form
\begin{equation}
\begin{array}{c}
\dfrac{\partial N }{\partial \theta} +
\dfrac{\partial }{\partial \rho}
\left(N\, V \right)
=0,\quad
\dfrac{\partial P }{\partial \theta} + E +
V \dfrac{\partial P}{\partial \rho} = 0, \vspace{1.5 ex}\\
\gamma = \sqrt{ 1+ P^2}, \quad
V = \dfrac{P}{\gamma},\quad
\dfrac{\partial E }{\partial \theta} = N\, V\,.
\end{array}
\label{3gl3}
\end{equation}
From the first and last equations (\ref {3gl3}) it follows
$$
\dfrac{\partial }{\partial \theta}
\left[ N +
\dfrac{\partial }{\partial \rho} E \right] = 0.
$$
This relation is valid both in the absence of plasma oscillations ($ N \equiv 1, \, E \equiv 0 $), and in their presence. Therefore, from here we have a simpler expression for electron density
\begin{equation}
 N = 1 -
\dfrac{\partial  E }{\partial \rho}.
\label{3gl4}
\end{equation}
Using it, we come to equations describing plane one-dimensional relativistic plasma oscillations:
\begin{equation}\label{u1}
\dfrac{\partial P }{\partial \theta}+
V\,\dfrac{\partial P}{\partial \rho}  + E = 0,\quad
\dfrac{\partial E }{\partial \theta}  +
V\,\dfrac{\partial E}{\partial \rho}- V = 0,
\quad V = \dfrac{P}{\sqrt{1+P^2}}\,.
\end{equation}

Here $ \rho $ and $ \theta $ are dimensionless coordinates in space and time, respectively. The variable $ P $ describes the momentum of electrons, $ V $ is the speed of electrons, $ E $ is a function that characterizes the electric field.
Similar statements were previously considered only in the physical literature    (see, e.g.,\cite{david72} and references therein) , where research methods and statements of results differ significantly from the present work.

 Below we study the Cauchy problem for \eqref {u1} with initial conditions
\begin{equation}\label{cd1}
   (P(\rho,0), E(\rho,0)) = ( P_0(\rho), E_0(\rho)) \in C^1 ({\mathbb R}), \rho \in {\mathbb R}.
\end{equation}
in the half-plane $\{(\rho,\theta)\,:\, \rho \in {\mathbb R},\; \theta
> 0\}$.
Moreover, we always consider the initial data with  derivatives uniformly bounded on ${\mathbb R}$.

One can also study  equation \eqref {u1} in a simplified (nonrelativistic) approximation. Namely, under the assumption that the electron velocity $ V $ is small, we have
$
   P= V + \dfrac{V^3}{2} + O(V^5),\,V\to 0.
$
Thus, up to cubic small terms, we can assume that $ P = V $. This assumption allows us to write \eqref {u1} as
\begin{equation}\label{u2}
\dfrac{\partial V }{\partial \theta}  +
V\,\dfrac{\partial V}{\partial \rho}+ E = 0,\quad
\dfrac{\partial E }{\partial \theta} +
V\,\dfrac{\partial E}{\partial \rho} - V
=0,
\end{equation}
while the initial conditions take the form
\begin{equation}\label{cd2}
     (V(\rho,0), E(\rho,0)) = ( V_0(\rho), E_0(\rho)) \in C^1 ({\mathbb R}), \rho \in {\mathbb R}.
\end{equation}

The systems \eqref {u1} and \eqref {u2} are hyperbolic. It is well known that for such systems there exists a local in time unique solution to the Cauchy problem \eqref{cd2} of the same class as the initial data, in our case it is $ C^1 $. It is also known that for such systems the loss of smoothness by a solution occurs in one of the following scenarios: either the solution components themselves go to infinity within a finite time, or they remain bounded, but their derivatives go to infinity \cite{Daf16}. Last opportunity
realized, for example, for homogeneous conservation laws, which include the equations of gas dynamics, where the appearance of a singularity corresponds to the formation of a shock wave. Thus, in order to prove the existence of a global global solution to the Cauchy problem, it is necessary to exclude each of these possibilities.

\section{Analysis of nonrelativistic case}\label{Sec2}

Let us start with  system \eqref {u2}, which is simpler. We write the equations of characteristics for it:
\begin{equation}\label{char2}
     \dfrac{dV}{d\theta}=-E,\quad \dfrac{dE}{d\theta}=V,\quad \dfrac{d\rho}{d\theta}=V.
     \end{equation}
It immediately imlpies $V^2+E^2= V^2(0)+E^2(0)=\rm const$ along each of the characteristics $\rho=\phi(\theta)$, starting from a point $(\rho_0,0)$. Therefore, the solution itself remains bounded and it remains to exclude the second possibility of a singularity formation.

We denote $v=V_\rho$, $e=E_\rho$ and differentiate the system \eqref{u2} with respect to $\rho$.  We get a pair of equations
\begin{equation}\label{char2d}
     \dfrac{dv}{d\theta}=-v^2-e,\quad \dfrac{de}{d\theta}=(1-e)v,
     \end{equation}
which together with \eqref {char2} form the extended system. We see that system \eqref{char2d} does not depend on \eqref{char2}, its phase trajectories on the plane $ (e, v) $ can be found as solutions of the Bernoulli equation
$$
     \dfrac{dv}{de}=\frac{v}{e-1}+ \frac{e}{(e-1)v}, \quad e\ne 1.
$$
This equation can be solved in a standard way, the first integral has the form
\begin{equation}\label{int2}
v^2+2 e -1 =C (e-1)^2.
\end{equation}
The resulting equation \eqref {int2} defines a second-order curve in the $ (e, v) $ - plane that is symmetric about the axis $ v = 0 $, its type depends on the sign of
$$
C=\frac{v^2+2 e -1}{ (e-1)^2},
$$
which coincides with the sign of the numerator
$$
\Delta=v^2+2 e -1.
$$

Namely, if $ \Delta <0 $, then the phase curves are bounded, they surround the singular point $ (0,0) $ (center), and the derivative solutions remain bounded throughout the entire time $ \theta> 0 $. Otherwise, the phase curve is a parabola for $ \Delta = 0 $ or a hyperbola for $ \Delta> 0 $, the derivatives of the solution along it become infinite.

Let us show that this happens within a finite time. Indeed,  \eqref {char2d}, \eqref {int2} implies that $ e $ satisfies the equation
$$
     \dfrac{de}{d\theta}=\pm (e-1) \sqrt{C (e-1)^2-2 e +1},
$$
     therefore
   \begin{equation}\label{razd2}
  \pm \int\limits_{e(\theta_*)}^{e(\theta^*)}  \dfrac{de}{(e-1) \sqrt{C (e-1)^2-2 e +1}}= \theta^*-\theta_*,
     \end{equation}
   where $ (\theta _ *, e (\theta_ *)) $ is the starting point of integration, and $ (\theta ^ *, e (\theta ^ *)) $ is the final one.
     Suppose the contrary, that is, the derivative $ e $ goes to infinity as $\theta \to \infty$. This means that $ \theta^* = \infty $, $ e (\theta ^ *) = \pm \infty $. Thus, the integral on the left side \eqref {razd2} must diverge. However, it is easy to see that this is not so. The resulting contradiction demonstrates that $ e $ (and therefore $ v $), goes to infinity in finite time.

To complete the proof, we note that for $ e = 1 $ the system \eqref {char2d} reduces to the equation
$$
 \dfrac{dv}{d\theta}=-v^2-1,
$$
which can be integratesd elementarily and whose solutions for any initial data turn into minus-infinity in a finite time; the value of $ \Delta $ is also positive in this case.

Using  formula \eqref {razd2}, we can find the period $ T $ of revolution along a closed phase trajectory for $ C <0 $. Namely,
$$
  \frac{T}{2}=\int\limits_{e_-}^{e_+}  \dfrac{de}{(e-1) \sqrt{C (e-1)^2-2 e +1}},\quad e_\pm=1+\frac{1\mp\sqrt{1+C}}{C}.
$$
   Calculating the integral, we get that $ T = 2 \pi $ for all $ C $. This fact is natural, since the $ 2 \pi $ -periodicity of a smooth solution already follows from the equations \eqref {char2}.

Let us summarize our reasoning as a theorem.

\begin{theorem} \label{T1}  For the existence and uniqueness of continuously differentiable  $ 2 \pi- $ periodic in time  solution $ V (\theta, \rho), \, E (\theta, \rho) $ of problem~\eqref {u2},
\eqref {cd2} is necessary and sufficient  that  inequality
\begin{equation} \label {crit2}
 \left (V'_0 (\rho) \right) ^ 2 + 2 \, E'_0 (\rho) -1 <0
\end {equation}
holds at each point $ \rho \in  \mathbb R$.

If there exists at least one point $ \rho_0 $ for which the inequality opposite to \eqref {crit2} holds, then the derivatives of the solution become infinite in a finite time.
\end{theorem}

 \medskip

In addition to this result, we note that functions of the form$V=W(\theta)\rho +
W_0(\theta)$, $E=D(\theta)\rho + D_0(\theta)$ are solutions of the equation \eqref {u2} if the coefficients  $W(\theta)$ and
$D(\theta)$, i.e. derivative solutions with respect to spatial
variable, obey  \eqref {char2d}. In\cite{CH18}, to construct approximate methods for solving \eqref {u2}, \eqref {cd2},  the  globally smooth solution with the following coefficients is used:
$$
  W(\theta) = \dfrac{s\, \cos (\theta + \theta_0)}{1 + s\, \sin (\theta + \theta_0)}, \quad
  D(\theta) = \dfrac{s\, \sin (\theta + \theta_0)}{1 + s\, \sin (\theta + \theta_0)},
 $$
$$
s = \dfrac{\sqrt{\alpha^2+\beta^2}}{1-\alpha},\;\;
\cos \theta_0 = \dfrac{\beta}{\sqrt{\alpha^2+\beta^2}}, \;\;
\sin \theta_0 = \dfrac{\alpha}{\sqrt{\alpha^2+\beta^2}}, \quad \alpha = D(0),\;\; \beta = W(0).
$$

\bigskip

\section{Analysis of the relativistic case}\label{Sec3}
We proceed to the study of  system  \eqref{u1}.  We write a system of characteristics for it:
\begin{equation}\label{char1}
     \dfrac{dP}{d\theta}=-E,\quad \dfrac{dE}{d\theta}=\dfrac{P}{\sqrt{1+P^2}},\quad \dfrac{d\rho}{d\theta}=\dfrac{P}{\sqrt{1+P^2}}.
     \end{equation}
It implies
\begin{equation}\label{first_int}
2\sqrt{1+P^2}+E^2=2\sqrt{1+P_0^2(\rho_0)}+E_0^2(\rho_0)=C_1(\rho_0)\ge 2
\end{equation}
along a characteristic starting from the point $ (\rho_0,0) $. Therefore, the solution itself remains bounded, as in the previous case.

  We also note that the estimate
$$
  1\le\sqrt{1+P^2}\le \frac{C_1(\rho_0)}{2}
$$
holds.

The period $ T (\rho_0) $ can be calculated as
$$
  {T (\rho_0)}=2\, \int\limits_{P_-}^{P_+}  \dfrac{dP}{ \sqrt{C_1 (\rho_0)-2\sqrt{1+P^2}}},\quad
  P_\pm=\pm\frac{\sqrt{C_1^2(\rho_0)-4}}{2},
$$
we omit the argument $ \rho_0 $ hereinafter for brevity. The period tends to $ 2 \pi $ at $ C_1 \to 2 $, but increases with $ C_1 $.
In contrast to the nonrelativistic case, when along each
characteristic the period is the same and equal to $ 2 \pi $, in the relativistic case the period on each characteristic is different.

     Let us build an extended system. To do this, we additionally denote $ p = P_\rho $ and differentiate \eqref{u1} with respect to $\rho$, and obtain a pair of equations
\begin{equation}\label{char1d}
     \dfrac{dp}{d\theta}=-e-\frac{p^2}{(1+P^2)^{3/2}},\quad \dfrac{de}{d\theta}=(1-e)\frac{p}{(1+P^2)^{3/2}},
     \end{equation}
We see that, unlike the previous case,  system \eqref{char1d} cannot be considered separately from \eqref{char1}. This increases its dimension and seriously complicates the analysis. It is necessary to attach the equation
\begin{equation}\label{char1d+}
     \dfrac{dP}{d\theta}=- E(P)=-\sqrt{C_1-2\sqrt{1+P^2}},
     \end{equation}
describing the behavior of the bounded function $ P $, the root can be extracted with both signs. The phase trajectory now belongs to three-dimensional space and the method used in the previous paragraph is not applicable here. System \eqref{char1d} can be considered as an inhomogeneous system of two equations with nonlinear resonant terms.
\bigskip

\subsection{Special case }\label{Sec3.1} We start with the case when the constant $ C_1 $ in  integral \eqref{first_int} is the same for all points $ \rho_0 \in \mathbb R $. In other words, the variable $ E $ can be expressed in terms of $ P $ all over the half-plane ${\mathbb R}\times{\mathbb R}_+$ and the system \eqref{u1} reduces to one equation
\begin{equation*}\label{u11}
\dfrac{\partial P }{\partial \theta}+
\dfrac{P}{\sqrt{1+P^2}}\,\dfrac{\partial P}{\partial \rho} = -E(P), \quad E(P)=\pm \sqrt{C_1-2\sqrt{1+P^2}},
\end{equation*}
for which a constant solution should be excluded.
 Since it follows from \eqref{first_int} that
\begin{equation}\label{pe}
e= - \frac{pP}{E(P)\sqrt{1+P^2} },\quad E=E(P),
\end{equation}
then the first equation \eqref{char1d} takes the form
\begin{equation}\label{p}
  \dfrac{dp}{d\theta}=\frac{pP}{E(P)\sqrt{1+P^2} }-\frac{p^2}{(1+P^2)^{3/2}}.
\end{equation}
Consider \eqref {p} together with \eqref{char1d+} and get a linear equation for the variable $ s = p^{- 1} $:
\begin{equation*}\label{sp}
  \dfrac{ds}{d P}=\frac{Ps}{E^2(P)\sqrt{1+P^2}}  -    \frac{1}{E(P) (1+P^2)^{3/2}}.
\end{equation*}
After integration  we obtain
\begin{equation}\label{pP}
  p= \frac{E(P)\sqrt{1+P^2}}{C_2\sqrt{1+P^2}-P},
\end{equation}
the constant $ C_2 $ is found from the initial data $ p_0 (\rho_0) =P'_0(\rho_0)$.
\begin{theorem}\label{T2}
Let the identity  $2\sqrt{1+P_0^2}+E_0^2\equiv \rm const$ hold with a constant that is independent of  $\rho_0\in \mathbb R$.
\begin{itemize}
\item If   for all  $\rho_0\in \mathbb R$ we have
\begin{equation}\label{C2}
C_2=\frac{E_0(\rho_0)}{P_0'(\rho_0)}+\frac{ P_0(\rho_0)}{\sqrt{1+P_0^2(\rho_0)}}\notin [-M_*,M_*],
\end{equation}
where $M_*=\frac{\sqrt{C^2_1-4}}{C_1}$,
   $ C_1 $ is defined in \eqref{first_int}, then the solution to problem \eqref{u1}, \eqref{cd1} remains smooth for all $ t> 0 $.

\item If the condition opposite to \eqref{C2} is satisfied at least for one $\rho_0$, then in a finite time the derivatives of solution become unbounded.
\end{itemize}
\end{theorem}
\begin{proof}  The derivative $ p $ is bounded over the entire range of $ P $, if the denominator in \eqref{pP} does not vanish for $ P^2 \le \frac {C_1^2}{ 4} -1 $.
The constant $ M_* $ is the maximum of the function $\frac{P}{\sqrt{1+P^2}}$ on the segment $[0,\sqrt{\frac{C_1^2}{4}-1}]$.
\end{proof}

It is easy to see that for sufficiently large values of $ p(\rho_0) $ the condition \eqref{C2} is certainly not satisfied.
\bigskip

\subsection{General case }\label{Sec3.2}
Denote by $K(\theta)=(1+P^2(\theta))^{-3/2}$ the periodic function found from \eqref{char1d+}, for which the  following estimate holds:
\begin{equation}\label{estK}
0<K_-=\frac{8}{C_1^3} \le K(\theta)\le 1.
\end{equation}

Introduce new variables $u={e}/{p}$, $\lambda=(1-e)/{p}$, in which the system \eqref{char1d} can be written quite simply:
\begin{equation}\label{supp1}
 \dfrac{du}{d\theta}=u^2+K,\quad \dfrac{d\lambda}{d\theta}=\lambda u.
\end{equation}

It immediately follows from the first equation \eqref{supp1} and \eqref{estK} that $ u $ goes to plus-infinity in a finite time, which can be estimated from above and from below. This can happen for two reasons: $ p $ vanishes at a finite $ e $ or $ p $ and $ e $ both go to infinity. Similarly, turning $ u $ to zero means that $ p $ goes to infinity or $ e $ goes to zero. Analysis of phase trajectories shows that in the region $ e \ge 1 $ the variable $ p $ always goes to infinity in a finite time. In the region $ e \in (0,1), \, p> 0 $, the trajectory reaches the boundary $ p = 0 $, and in the regions $ e \in (0,1), \, p <0 $ and $ e <0, \, p> 0 $, the trajectory reaches  the boundary $ e = 0 $. In the quadrant $ e <0, \, p <0 $, where $ u> 0 $, the value $ u $
  can go to infinity in different ways. If $ p $ vanishes, then the trajectory makes a revolution around the origin, otherwise $ p $ goes to infinity in a finite time. We need to distinguish between these situations.

For this, we note that ${p}^{-1} = \lambda+u $ and this expression can be evaluated on both sides as
$$
\psi_-(\theta)\le \frac{1}{p}\le \psi_+(\theta),
$$
where
$$
 \psi_-(\theta)=\sqrt{K_-}\tan \left(\sqrt{K_-}\theta+\arctan \frac{u_0}{\sqrt{K_-}}\right)+  \frac{\lambda_0\sqrt{1+\tan^2\left(\sqrt{K_-}\theta+\arctan \frac{u_0}{\sqrt{K_-}}\right)}}{\sqrt{\frac{u^2_0}{K_-}+1}},
$$
$$
\psi_+(\theta)=\tan \left(\theta+\arctan {u_0}\right)+\frac{\lambda_0\sqrt{1+\tan^2\left(\theta+\arctan {u_0}\right)}}{\sqrt{u_0^2+1}}.
$$
If in the domain of non-positivity the value $ {p}^{- 1} $ vanishes before it goes to infinity, then the solution loses smoothness in a finite time. We must find the conditions under which  $ \psi_- (\theta) $ vanishes till the end of the period of $ \psi_- (\theta) $ as $\theta$ increases. This happens if
$\lambda_0^2\le u_0^2 + K_-,$ that is, $$ K_- p_0^2\ge 1-2 e_0. $$
If $\lambda_0^2> u_0^2 + 1,$ then
  $ \psi _ + (\theta)\ne 0 $ till the end of the period $\pi$ of $ \psi_+ (\theta) $. Thus, if
$$ p_0^2 <1-2 e_0, \, p_0<0, $$
then the trajectory does not go to infinity in the lower half-plane, but passes into the upper half-plane and makes a half-turn there.
This condition coincides with \eqref {crit2} for negative $p_0$. In this case, \eqref {crit2} does not provide a sufficient condition for maintaining global smoothness, but only means that the phase trajectory makes at least one revolution around the origin.

If $ p_0 \ge 0 $, we can only guarantee that the trajectory will fall into the lower half-plane, but the condition $ p^2 <1-2 e $ may not be satisfied there.

Note that if
$2\sqrt{1+P^2}+E^2\equiv \rm const$,  condition \eqref {pe} signifies that the first equation \eqref {supp1} is a consequence of \eqref{char1d+}, and $ u $ goes to infinity or zero when $ E $ or $ P $ go to zero, respectively.

\medskip

 So, our result is as follows.
\begin{theorem}\label{T3}
Let  the expression
 $2\sqrt{1+P_0^2}+E_0^2$ be not equal to constant identically. Then if there exists at least one point $ \rho_0 $ for which the inequality
$$
K_-\left(P'_0(\rho_0)\right)^2 + 2\, E'_0(\rho_0)-1\ge 0,\quad K_-=\frac{8}{(2\sqrt{1+P^2(\rho_0)}+E^2(\rho_0))^3},
$$
holds,  then the derivatives of solution to \eqref{u1}, \eqref{cd1} become infinite within a finite time not exceeding the period of oscillation $ T (\rho_0) $.

If for all $ \rho $ condition
$$
 \left(P'_0(\rho)\right)^2 + 2\, E'_0(\rho)-1< 0, \quad P'_0(\rho)<0,
$$
holds, this ensures  smoothness of the solution for at least a period $ T (\rho) $ on each characteristic, that is, up to time $t_*=\inf\limits_{\rho\in \mathbb R} T(\rho)$, $t_*>2\pi$.

\end{theorem}

Since  system \eqref {char1d} is not autonomous, at the next turnover the trajectory may already go into the domain ensuring a singularity formation. Numerical experiments show that this is exactly what is happening. Apparently, the trajectories of the solution \eqref {char1d} are moving away from the origin due to nonlinear resonance. To prove this fact, estimates alone are not enough, and an explicit form of the function $ K $ is required.


We note that the following theorem is also valid, which connects the problem of the appearance of a singularity of the system of plasma oscillations with the theory of linear equations with periodic coefficients.
\begin{theorem}\label{T4}
Let the expression $2\sqrt{1+P_0^2}+E_0^2$ be not equal to constant identically.  We denote
$$
\dfrac{E'_0(\rho_0)}{P'_0(\rho_0)} = u_0, \quad
\dfrac{1-E'_0(\rho_0)}{P'_0(\rho_0)} = \lambda_0
$$ under the assumption that
$P'_0(\rho_0)\ne 0$. Let $z(\theta)$ be the solution to the Cauchy problem for the Hill equation
\begin{equation}\label{Hill}
\dfrac{d^2 z}{d\theta^2}+K(\theta) z=0,\quad z(0)=1,\, z'(0)=-u_0.
\end{equation}
If for at least one point $ \rho_0 $ there exists a moment of time
$\theta_*>0$ such that $z'(\theta_*)=\lambda_0$, then the derivatives of solution to \eqref {u1}, \eqref {cd1} become infinite in a finite time.

Otherwise, the solution remains smooth for all $ \theta> 0 $.
\end{theorem}

\begin{proof}  The substitution $ u = -z '/ z $ reduces the first of equations \eqref{supp1} to the Hill equation \eqref{Hill}. Due to the homogeneity of the latter, we can set $ z (0) = 1 $, then $ z '(0) = - u_0 $. According to the second of the equations \eqref {supp1} we have $u=\lambda'/\lambda$,
therefore $\lambda=\lambda_0/z$. Then
$$
\dfrac{u}{\lambda} = \dfrac{e}{1-e} = -\,\dfrac{z'}{\lambda_0}.
$$
If $e$
turns into infinity at
 $\theta=\theta_*$, then
$-\dfrac{z'(\theta_*)}{\lambda_0}=-1$, i.e.
$z'(\theta_*)=\lambda_0$. If there is no such moment of time, then $ e $, and
 $ p $ remain bounded. The theorem is proved.
\end{proof}


 It is easy to verify that for $ K = 1 $, that is, for the case of nonrelativistic oscillations, we obtain   criterion \eqref {crit2} as a corollary.

\medskip
In order to use Theorem \ref{T4}, we also need the explicit form of $ P $. We can get it by making an assumption about the smallness of oscillations. Namely, the following corollary holds.
\begin{corollary}\label{С1}
Any solution to the Cauchy problem \eqref {u1}, \eqref {cd1} which is an arbitrarily small deviation from
equilibrium state $ P = 0 $, $ E = 0 $ and for which $ 2 \sqrt {1 + P_0 ^ 2} + E_0 ^ 2 $ is not equal to a constant identically, the derivatives of solution become infinite in a finite time.
\end{corollary}

\begin{proof}
If  $ C_1 = 2 + \epsilon^2 $, $ \epsilon \ll 1 $, then  $ P $ and $ E $ are small at any time, provided that the solution remains smooth. Let us choose $ P_0 (\rho_0) $ and $ E_0 (\rho_0) $ sufficiently small and get $ P^2 \le \epsilon^2 \ll 1 $. We expand the expression under the root in \eqref{char1d+} in a series in $ P $ and keep only the quadratic terms. Thus, we get $ P = \epsilon \sin (\theta + \theta_0) $. Without loss of generality, put $ \theta_0 = 0 $, substitute the expression for $ P $ into $ K $ and expand the result in a series in $ \epsilon $. We get
\begin{equation}\label{K_series}
K(\theta)=a-2b \cos 2\theta + O(\epsilon^4),\quad a=1-
\frac{3}{4}\epsilon^2,\quad b= -\frac{3}{8}\epsilon^2.
\end{equation}
If we neglect the terms of order greater than the second in \eqref{K_series} and substitute it  in \eqref {Hill}, we obtain the Mathieu equation, whose theory is well developed (for example,\cite {BEr67}, chapter 16). According to Floquet theory,
boundedness or unboundedness of solution of such an equation together with its derivative is completely determined by its characteristic exponent $ \mu $, defined as solution of the equation $\cosh \mu
\pi = z_1(\pi),$ where $z_1(\theta)$ is the solution to the Mathieu equation
subject to initial conditions $z_1(0)=1$, $z_1'(0)=0$, i.e. one of Mathieu functions. Unboundedness occurs for real $ \mu $, that is, if $|\cosh \mu \pi|>1$. According to the asymptotic formula \cite {BEr67}, Sec.16.3 (2), which
(taking into account the typo in the sign in the formula of Sec.16.2 (15))
has the form $$\cosh \mu \pi = \cos
\sqrt a \pi-\frac{\pi b^2}{(1-a)\sqrt a }\sin \sqrt a \pi
+O(b^4),\,b\to 0,$$ we get
$$\cosh \mu
\pi=-1-\frac{27}{2048}\pi^2\epsilon^6+O(\epsilon^8)<-1.
$$
In addition, since $ K> 0 $, any solution to the Mathieu equation is oscillating. Thus, for some finite time, $ z '(\theta) $ will reach any $ \lambda_0 $. The corollary is proved.
\end{proof}

\section{Traveling waves}\label{Sec4}
Physicists have long known solutions in the form of a wave traveling with a constant  velocity (quasistationary solutions) for cold plasma equations, including for the relativistic case.  An approximate analysis was performed in \cite {AL51,AP56}. In particular, it was known that the speed of a traveling wave is related to its smoothness and to the energy that the wave possesses. But we are interested in traveling waves primarily because they are the simplest example of solutions for which the first integral is constant in the whole space. Therefore, they can serve as an example of application of Theorem \ref{T2} and a counterexample showing that Corollary \ref{T4} is false without additional restrictions on the initial data. In addition, we will be interested in the type of singularity that the wave possesses in the absence of smoothness.

So, suppose that $P(\theta, \rho)={\mathcal P}(\xi)$
$E(\theta, \rho)={\mathcal E}(\xi)$, with a self-similar variable $\xi=\rho-w \theta$, $w=\rm const$.
Equation  \eqref{u1}  implies that functions
 ${\mathcal P}$ and ${\mathcal E }$ satisfy the autonomous system of ordinary differential equations
\begin{equation}\label{wave_sys}
(-w +\frac{{\mathcal P}}{\sqrt{1+{\mathcal P^2}}}){\mathcal P}'=-{\mathcal E},\qquad
(-w +\frac{{\mathcal P}}{\sqrt{1+{\mathcal P^2}}}){\mathcal E}'=\frac{{\mathcal P}}{\sqrt{1+{\mathcal P^2}}}.
\end{equation}
The system has a first integral
\begin{equation}\label{first_int1}
2(\sqrt{1+{\mathcal P}^2}-1)+{\mathcal E}^2={\mathcal I}^2=2(\sqrt{1+{\mathcal P}^2(0)}-1)+{\mathcal E}^2(0)=\rm const,
\end{equation}
similar to \eqref {first_int}, but which is constant not only along a specific characteristic. This means that we are in the situation of the special case, discussed in Sec.\ref{Sec3.1}. Thus, setting $ {\mathcal P} (0), {\mathcal E} (0) $, we obtain
\begin{equation}\label{wave_E}
{\mathcal E}=\pm\sqrt{{\mathcal I}^2-2(\sqrt{1+{\mathcal P}^2}-1)}.
\end{equation}
From \eqref{wave_sys}, \eqref{first_int1} we get the equation for determining the profile of a traveling wave
\begin{equation}\label{wave_P}
{\mathcal P}'   = \frac{\pm\sqrt{{\mathcal I}^2-2(\sqrt{1+{\mathcal P}^2}-1)}}{-w +\frac{{\mathcal P}}{\sqrt{1+{\mathcal P^2}}}}.
\end{equation}
Since $${\mathcal P}^2\le (\frac{{\mathcal I}^2}{2}+1)^2-1 (>0),$$
then under the condition
\begin{equation}\label{w}
w^2> 1-\frac{4}{({\mathcal I}^2+2)^2}
\end{equation}
the denominator in \eqref {wave_P} preserves its sign and the function $ {\mathcal P} $ at each half-period of $ \xi $ increases or decreases, changing between the values ${\mathcal P}_\pm= \pm \sqrt{(\frac{{\mathcal I}^2}{2}+1)^2-1}$.
The period $ X $ can be found as
\begin{equation*}\label{wave_X}
X   =2 \,{\rm sign} \,w
 \int\limits_{{\mathcal P}_-}^{{\mathcal P}_+} \frac{-w +\frac{{\mathcal P}}{\sqrt{1+{\mathcal P^2}}}}{\sqrt{{\mathcal I}^2-2(\sqrt{1+{\mathcal P}^2}-1)}} \, d \mathcal P.
\end{equation*}
\begin{center}
\begin{figure}[htb]
\begin{minipage}{0.495\columnwidth}
\includegraphics[width=2\columnwidth]{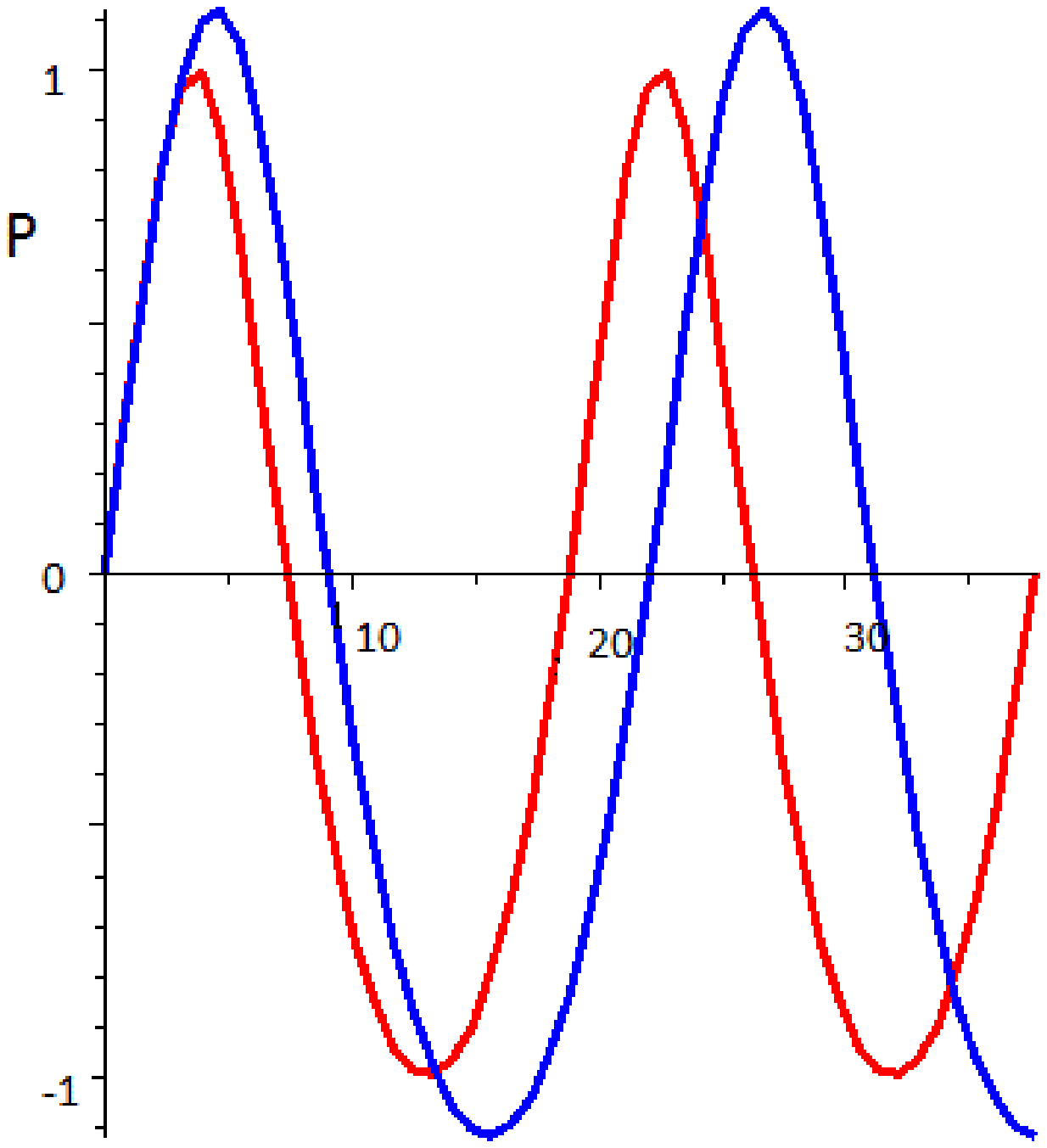}
\end{minipage}
\begin{minipage}{0.495\columnwidth}
\includegraphics[width=2\columnwidth]{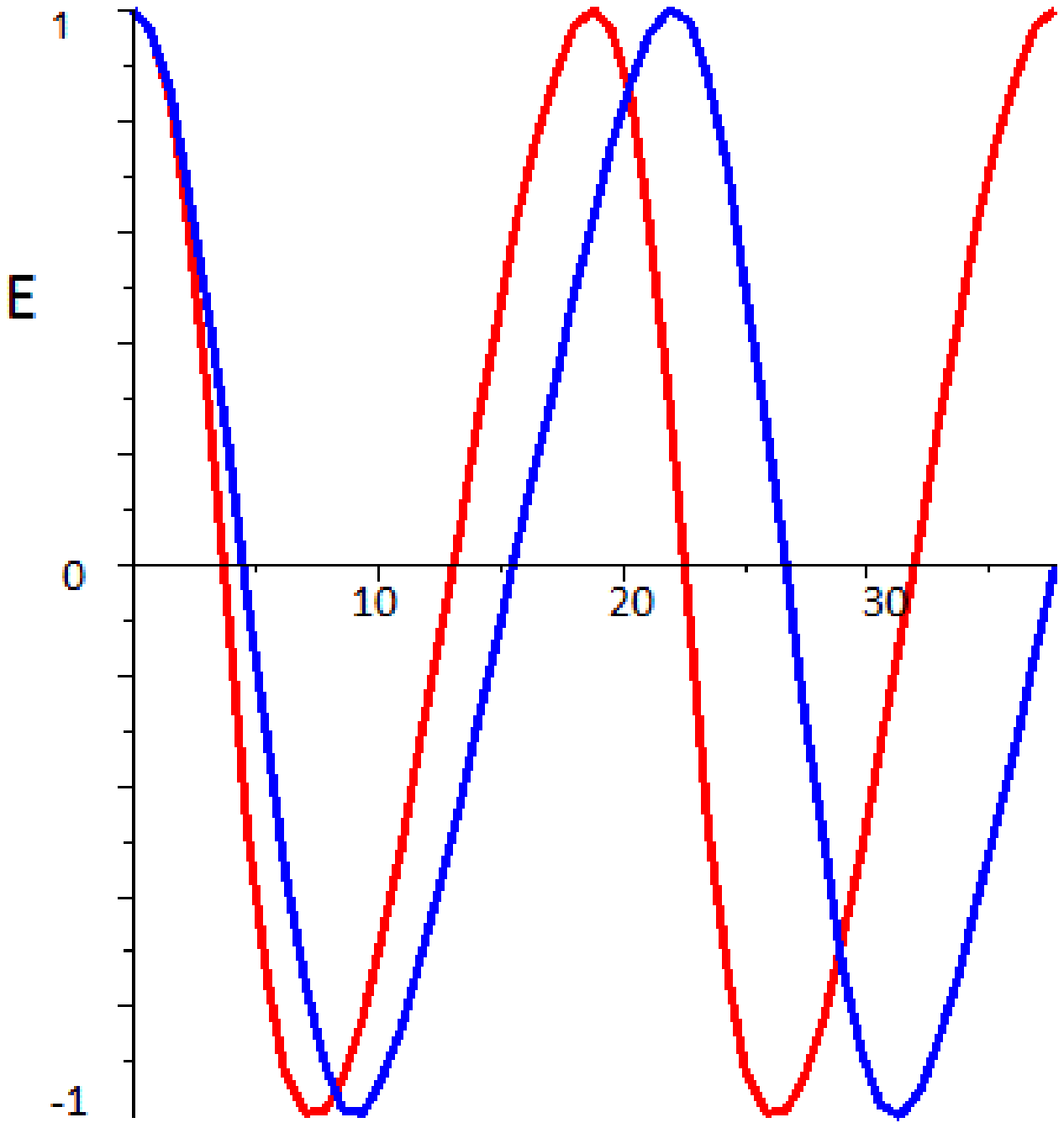}
\end{minipage}
\caption{Traveling wave profile for $ w = 3 $ at $ {\mathcal P} (0) = 0 $, $ {\mathcal E} (0) = 1 $. Two nonrelativistic periods in space, nonrelativistic and relativistic cases, red and blue graphs respectively.
Right: ${\mathcal P}(\xi)$. Left: ${\mathcal E}(\xi)$.}\label{Pic1}
\end{figure}
\end{center}
\begin{center}
\begin{figure}[htb]
\includegraphics[scale=0.33]{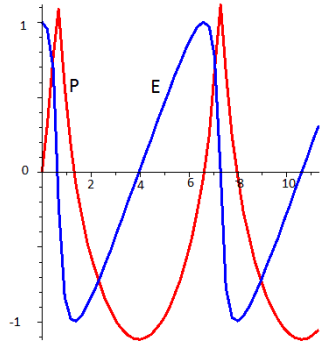}
\caption{Traveling wave profile at $ {\mathcal P} (0) = 0 $, $ {\mathcal E} (0) = 1 $ for $ w = 0.9 $ (the parameter is close to the critical one), relativistic case, $ {\mathcal P} (\xi) $ and $ {\mathcal E} (\xi) $, red and blue graphs respectively.}\label{Pic2}
\end{figure}
\end{center}
Note that if we choose $ w $ that does not satisfy the condition \eqref{w}, then we get a traveling wave with singularities that are similar to the emerging singularities of  smooth solutions obtained in numerical experiments (see the next paragraph, Fig.\ref{Pic4}), that is, the velocity component initially has a weak discontinuity, while the component of the electric field strength is like a step.

A traveling wave can also be constructed for the nonrelativistic case. In this case, the analogue of the equation \eqref{wave_P} can be explicitly integrated, as a result we get an implicit representation of the profile of the traveling wave  for $V(\theta, \rho)={\mathcal V}(\xi)$, namely,
\begin{equation*}\label{wave V}
 \xi+c= w \arcsin \frac{\mathcal V}{\mathcal I}_0 + \sqrt{{\mathcal I}_0^2-{\mathcal V}^2},
  \end{equation*}
$c= w \arcsin \frac{{\mathcal V}(0)}{{\mathcal I}_0} + \sqrt{{\mathcal I}_0^2-{\mathcal V}(0)^2}$, ${\mathcal I}_0^2={\mathcal V}^2(0)+{\mathcal E}^2(0)$. Moreover, the period with respect to the spatial variable can be easily calculated, $X=2\pi |w|$.

The condition for the wave speed to guarantee its smoothness in the nonrelativistic case is $ w^2> {\mathcal I}_0^2 $ and it is more stringent than \eqref{w}, because it dictates the growth of $ w $ along with the growth of "full energy" $ {\mathcal I}_0 $. In the relativistic case, the choice of $ w^2> 1 $ is sufficient to construct a smooth traveling wave.

 Fig.\ref{Pic1} shows the initial data profiles for a smooth traveling wave in the relativistic and nonrelativistic cases. Fig.\ref{Pic2} presents the process of forming a singularity as the parameter $ w $ approaches the critical value for the relativistic case. For the nonrelativistic case, with this value of the parameter, there is no longer a smooth solution. We see that the velocity component $ {\mathcal P} $ forms a weak singularity, while the component $ {\mathcal E} $ forms a step. The singularities of numerical solution in Fig.\ref{Pic4} have the same nature.

It should also be noted that choosing small $ {\mathcal I} $, we can construct a family of small perturbations of the trivial state, depending on the parameter $ w $, which determines the period of the wave in space. For sufficiently large $ w^2 $, this family is an infinitely smooth function in time. However, according to Corollary of Theorem 4, a small perturbation of the initial data in the general case leads to a solution for which smoothness is lost. Thus, small perturbations in the form of traveling waves are unstable in the class of all small perturbations.

Numerical results also suggest that all traveling waves, not only small, as well as other globally smooth solutions, the existence of which is guaranteed by Theorem \ref{T2} under a special condition on the data, can be also destroyed by infinitely small perturbations. The hypothesis is that in the relativistic case, the set of smooth initial data corresponding to a globally in time smooth solution has measure zero in the space of all smooth initial data. In practice, this means that globally smooth solutions are not observed in the relativistic case.


\section{Numerical experiments}\label{Sec5}
As the initial conditions \eqref{cd1}, we choose
\begin{equation}\label{gauss}
E_0(\rho) = \left(\dfrac{a_*}{\rho_*}\right)^2\rho \exp\left\{-2
\dfrac{\rho^2}{\rho_*^2}\right\},\quad P_{0}(\rho) = 0.
\end{equation}
According to \eqref {gauss} and \eqref {3gl4}, initially the maximum of density is at the origin. This choice is the most natural from the point of view of physics, since it simulates the effect on a rarefied plasma of a short powerful laser pulse when it is focused in a line (this can be achieved using a cylindrical lens), see \cite {Shep13} for details.

In \cite {CH18} it is shown that by selecting the parameters $ \rho_* $ and $ a_* $ in the initial conditions,
it is possible to obtain the singularity of electron density at an arbitrary moment in time when the electrons move from the origin.

The parameters characterize the scale of the localization domain and the maximum  of the electric field  $ E_ {\max} = a _*^2 / (\rho_* 2 \sqrt {{\rm e}}) \approx 0.3 a_*^2 / \rho_* $,
here $ {\rm e}  $ is the base of the natural logarithm.

\begin{figure}[h]
\includegraphics[scale=0.8]{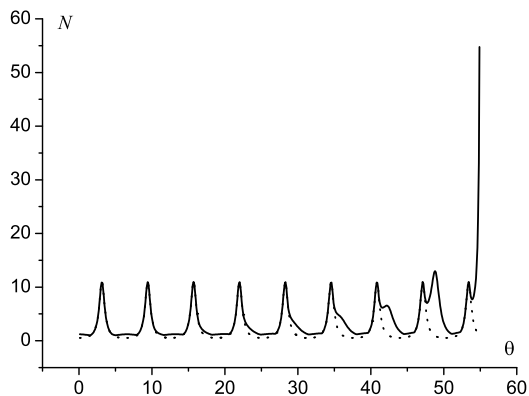}
\vspace{-0.5 cm}
\caption{
Dynamics of electron density: maximum in the region (solid line) and at the origin (dashed line).}\label{Pic3}
\end{figure}
\medskip
\medskip

 To demonstrate relativistic breakdown of oscillations, we choose the parameter values $ a_* = 2.07, \; \rho_* = 3.0 $ and note that, depending on their choice, the concentration of electrons in the center of the region can many times exceed the equilibrium (background) value equal to 1. These parameters lead to fluctuations of low intensity, when the amplitude of the oscillations is only about 10 times higher than the background value.

Since the initial conditions, by virtue of Theorem \ref{T3}, ensure the existence of a solution for more than one period, two trends can be noted in accordance with the theoretical analysis in the process of oscillations. The first of these is that the density fluctuations outside the coordinate origin are somewhat ahead in phase of the density fluctuations at the point $ \rho = 0 $ and this phase shift increases from period to period. The second trend is more obvious: over time, a gradual formation of the absolute maximum of density, located outside the origin, occurs.

 In Fig.\ref{Pic3}, the dotted line shows the change in time of electron density at the origin, and the solid line indicates the dynamics of the maximum value over the region.
At first, the oscillations are regular, i.e. global maxima and minima of density in the region succeed each other after half the period and are located at the origin. After the seventh regular (central) maximum at time moment $ \theta \approx 42.2 $, a new structure arises, it is the maximum electron density outside the origin, while regular oscillations continue to be observed in the vicinity of the origin. The newly arising maximum, at the moment $ \theta \approx 48.8 $, increases approximately twice in size and on the next period increases at $ \theta \approx 55.1 $. In its place a delta-shaped singularity of electron density arises.

\begin{center}
\begin{figure}[h]
\includegraphics[scale=0.8]{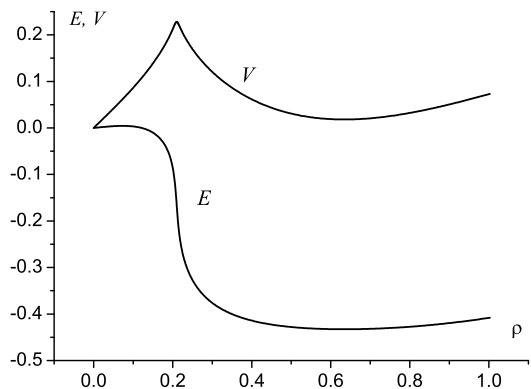}
\vspace{-0.5 cm}
\caption{
Spatial distribution of velocity and electric field at the moment of formation of the second off-axis maximum.}\label{Pic4}

\end{figure}
\end{center}

Fig.\ref{Pic4} shows the spatial distributions of the velocity $ V $ and the electric field $ E $ at the moment $ \theta \approx 48.8 $, when the absolute maximum outside the coordinate origin was completely formed.

Note that, due to the structure of equations \eqref {u1}, its solution will remain an odd  function of $\rho $  if the initial data have this property. The initial data \eqref {gauss} is just that, so we will present the results of numerical calculations only on the positive axis.

We see that in the vicinity of the density maximum the velocity component forms a discontinuity of the derivative (weak discontinuity), but not of the function itself, while the electric field function forms a strong discontinuity. It is precisely such qualitative characteristics of $ V $ and $ E $ that provide the breakdown of oscillations at the moment $ \theta \approx 55.1 $. It is important to note that the breakdown has a character of a `` gradient catastrophe'', i.e. the functions $ V $ and $ E $ themselves remain bounded.

Numerical experiments show that in the relativistic case formation of singularities is the rule.

\section*{Conclusion and discussion}

Let us dwell on the applicability of the results obtained. First note that the electron density  $ N (\rho, \theta) $ should be nonnegative in the process of oscillation. The Gauss theorem \cite {D,david72}, which relates the electron density to a function that describes the electric field in the dimensionless variables, has the form
$$
\dfrac{\partial E(\rho,\theta)}{\partial \rho} = 1 -  N(\rho,\theta).
$$
In turn,  \eqref {crit2} implies that the inequality $ 1 - 2 \, e - p^2> 0 $ persists in time; therefore, for an arbitrary $ \theta> 0 $ we have the previously obtained  lower bound for the electron density $ N (\rho, \theta)> 1/2 $ on an infinite oscillation time interval.
(see  \cite{david72} for a different derivation). This property is important for the theoretical justification of approximate
methods for modeling plasma oscillations.

In addition, we note that condition \eqref {crit2} allows us to formulate the generalized Riemann problem
(with piecewise linear initial conditions) for equations \eqref {u2}, the solution of which does not have a strong singularity of the electron density function. The practical benefit of solving the generalized Riemann problem is the construction of numerical scheme of the second order of accuracy in time and space for the Cauchy problem considered here \cite{KPS12}.

Finally, as follows  from our results,  approximate methods for problem \eqref{u1}, \eqref{cd1},  must take into account the possibility of the effect of breaking of oscillations as a completion of calculations. In particular, this means the necessity  for  refinement of the discretization parameters when using schemes in Euler variables. Thus, the results of this work are of fundamental importance for the development of the theory of computational methods and practical calculations of problems on oscillatory motions of a cold plasma.

Note that there is a very large number of works devoted to the analysis of solutions to the Cauchy problem for hyperbolic systems, which are various versions of models of magnetohydrodynamics with relativistic effects, for which there are results on the finite-time formation of a singularity (a survey can be found, for example, in \cite{Ren05}). In addition, processes in a collisionless plasma can be modeled using the Vlasov-Poisson equations, that is, based on the kinetic approach. For such models, the question of the conditions for the existence of a smooth solution and the possibility of the formation of singularities was also investigated (see the review \cite{Skub14} and the references). Based on the kinetic approach, the  traveling wave solutions (Bernstein-Green-Kruskal modes) are also known.

However, there are extremely few models in which the criterion for the formation of singularities in terms of initial data can be obtained, and the nature of this singularity can be theoretically analyzed. In the present work, we study one of these models, which is very informative and at the same time relatively simple.

We can say that the main practical conclusion of this paper is that in the general case, all non-trivial relativistic oscillations break down with time.

\section*{Acknowledgment}
Partially supported by the Moscow Center for Fundamental and Applied Mathematics.


\end{document}